\documentclass[11pt]{article}

\title{Extending the Reach of the Point-to-Set Principle}

\author{Jack H. Lutz\footnote{Research supported in part by National Science Foundation grants 1545028 and 1900716 and partly done during visits at the Institute for Mathematical Sciences at the National University of Singapore and the California Institute of Technology.}\\Iowa State University \and Neil Lutz\footnote{Part of this work was done while this author was at the University of Pennsylvania.}\\Iowa State University \and Elvira Mayordomo\footnote{{Research supported in part by Spanish Ministry of Science, Innovation and Universities grants TIN2016-80347-R and PID2019-104358RB-I00 and partly done during a visit to the Institute for Mathematical Sciences at the National University of Singapore.}}\\Universidad de Zaragoza}

\date{}

\usepackage[margin=1in]{geometry}
\usepackage{amsfonts}
\usepackage{amsmath}
\usepackage{amsthm}
\usepackage{amssymb}
\usepackage{graphicx}
\usepackage{mathtools}
\usepackage{url}
\usepackage[hidelinks]{hyperref}
\usepackage{enumerate}

\theoremstyle{plain}
\newtheorem{theorem}{Theorem}
\numberwithin{theorem}{section}
\newtheorem{lemma}[theorem]{Lemma}
\newtheorem{corollary}[theorem]{Corollary}
\newtheorem{observation}[theorem]{Observation}

\theoremstyle{definition}
\newtheorem{construction}[theorem]{Construction}
\newtheorem*{defn}{Definition}

\numberwithin{equation}{section}

\newcommand{\dimm}{\underline{\dim}_{\mathcal{M}}}
\newcommand{\Dimm}{\overline{\dim}_{\mathcal{M}}}
\newcommand{\cK}{\mathcal{K}}
\newcommand{\dimpack}{{\dim}_{\mathrm{P}}}
\newcommand{\dimh}{\mathrm{dim}_\mathrm{H}}
\newcommand{\Dim}{{\mathrm{Dim}}}
\newcommand{\diam}{\mathrm{diam}}
\newcommand{\R}{\mathbb{R}}
\newcommand{\N}{\mathbb{N}}
\newcommand{\Z}{\mathbb{Z}}
\newcommand{\Q}{\mathbb{Q}}
\newcommand{\Rn}{\R^n}
\newcommand{\myset}[2]{ \left\{ #1 \left| \, #2 \right. \right\} }
\newcommand{\range}{\mathrm{range}}
\newcommand{\binary}{\{0,1\}^*}

\newcommand{\C}{{\mathrm{C}}}
\newcommand{\rhoH}{\rho_\mathrm{H}}
\newcommand{\cD}{\mathcal{D}}
\newcommand{\bin}{\{0,1\}}

\begin{document}
	\maketitle

\begin{abstract}
The \emph{point-to-set principle} of J. Lutz and N. Lutz (2018) has recently enabled the theory of computing to be used to answer open questions about fractal geometry in Euclidean spaces $\R^n$.  These are classical questions, meaning that their statements do not involve computation or related aspects of logic.

In this paper we extend the \emph{reach} of the point-to-set principle from Euclidean spaces to arbitrary separable metric spaces $X$.  We first extend two fractal dimensions--—computability-theoretic versions of classical Hausdorff and packing dimensions that assign dimensions $\dim(x)$ and $\Dim(x)$ to \emph{individual points} $x\in X$---to arbitrary separable metric spaces and to arbitrary gauge families.  Our first two main results then extend the point-to-set principle to arbitrary separable metric spaces and to a large class of gauge families. 

We demonstrate the power of our extended point-to-set principle by using it to prove new theorems about classical fractal dimensions in hyperspaces.  (For a concrete computational example, the stages $E_0, E_1, E_2, \ldots$ used to construct a self-similar fractal $E$ in the plane are elements of the hyperspace of the plane, and they converge to $E$ in the hyperspace.) Our third main result, proven via our extended point-to-set principle, states that, under a wide variety of gauge families, the classical packing dimension agrees with the classical upper Minkowski dimension on \emph{all} hyperspaces of compact sets.  We use this theorem to give, for all sets $E$ that are analytic, i.e., $\mathbf{\Sigma}^1_1$, a tight bound on the packing dimension of the hyperspace of $E$ in terms of the packing dimension of $E$ itself.
\end{abstract}

\section{Introduction}\label{sec:intro}

It is rare for the theory of computing to be used to answer open mathematical questions---especially questions in continuous mathematics---whose statements do not involve computation or related aspects of logic.\footnote{We use the adjective ``classical'' for theorems and questions whose statements do not involve computability or logic, regardless of when they were proven or formulated. A ``classical'' theorem can thus be very new.} The {\it point-to-set principle}~\cite{LutLut18}, described below, has enabled several recent developments that do exactly this. This principle has been used to obtain strengthened lower bounds on the Hausdorff dimensions of generalized Furstenberg sets~\cite{LutStu20}, extend the fractal intersection formula for Hausdorff dimension from Borel sets to arbitrary sets~\cite{Lutz21}, and prove that Marstrand's projection theorem for Hausdorff dimension holds for any set $E$ whose Hausdorff and packing dimensions coincide, whether or not $E$ is analytic~\cite{LutStu18}.\footnote{These very non-classical proofs of new classical theorems have provoked new work in the fractal geometry community. Orponen~\cite{Orpo20} has very recently used a discretized potential-theoretic method of Kaufman~\cite{Kauf68} and tools of Katz and Tao~\cite{KatTao01} to give a new, classical proof of the two main theorems of~\cite{LutStu18}.} (See~\cite{DowHir19a,Dowhir19b,LutLut20,LutMay21} for reviews of these developments.) These applications of the point-to-set principle all concern fractal geometry in Euclidean spaces $\R^n$.\footnote{Applications of the theory of computing---specifically Kolmogorov complexity---to discrete mathematics are more numerous and are surveyed in~\cite{LiVit19}. Other applications to continuous mathematics, not involving the point-to-set principle, include theorems in descriptive set theory~\cite{oMosc80,HjKeAl97,KeSoTo99}, Riemannian moduli space~\cite{Weinbe04}, and Banach spaces~\cite{KihPau14}.}

This paper extends the reach of the point-to-set principle beyond Euclidean spaces. To explain this, we first review the point-to-set principle to date. (All quantities defined in this intuitive discussion are defined precisely later in the paper.) The two best-behaved classical fractal dimensions, Hausdorff dimension and packing dimension, assign to every subset $E$ of a Euclidean space $\R^n$ dimensions $\dimh(E)$ and $\dimpack(E)$, respectively. When $E$ is a ``smooth'' set that intuitively has some integral dimension between 0 and $n$, the Hausdorff and packing dimensions agree with this intuition, but more complex sets $E$ may have any real-valued dimensions satisfying $0\leq \dimh(E) \leq \dimpack(E)\leq n$. Hausdorff and packing dimensions have many applications in information theory, dynamical systems, and other areas of science \cite{Bill65,Falc14,KatHas95,Pesin98}.

Early in this century, algorithmic versions of Hausdorff and packing dimensions were developed to quantify the information densities of various types of data. The computational resources allotted to these algorithmic dimensions range from finite-state to computable enumerability and beyond, but the point-to-set principle concerns the computably enumerable algorithmic dimensions introduced in~\cite{DISS,ESDAICC}.\footnote{These have also been called ``constructive'' dimensions and ``effective'' dimensions by various authors.} These assign to each {\it individual point} $x$ in a Euclidean space $\R^n$ an {\it algorithmic dimension} $\dim(x)$ and a {\it strong algorithmic dimension} $\Dim(x)$. The point-to-set principle of~\cite{LutLut18} is a complete characterization of the classical Hausdorff and packing dimensions in terms of oracle relativizations of these very non-classical dimensions of individual points. Specifically, the point-to-set principle says that, for every set $E$ in a Euclidean space $\R^n$,
\begin{equation}\label{eq:p2sh}
	\dimh(E)=\adjustlimits\min_{A\subseteq\N}\sup_{x\in E}\,\dim^A(x)
\end{equation}
and
\begin{equation}\label{eq:p2sp}
	\dimpack(E)=\adjustlimits\min_{A\subseteq\N}\sup_{x\in E}\,\Dim^A(x),
\end{equation}
where the the dimensions on the right are relative to the oracle $A$. The point-to-set principle is so named because it enables one to use a lower bound on the relativized algorithmic dimension of a single, judiciously chosen \emph{point} in a set $E$ to prove a lower bound on the classical dimension of the \emph{set} $E$.

The classical Hausdorff and packing dimensions work not only in Euclidean spaces, but in arbitrary metric spaces. In contrast, nearly all work on algorithmic dimensions to date (the exception being~\cite{edgms}) has been in Euclidean spaces or in spaces of infinite sequences over finite alphabets. Our objective here is to significantly reduce this gap by extending the theory of algorithmic dimensions, along with the point-to-set principle, to arbitrary separable metric spaces. (A metric space $X$ is \emph{separable} if it has a countable subset $D$ that is {\it dense} in the sense that every point in $X$ has points in $D$ arbitrarily close to it.) 

In parallel with extending algorithmic dimensions to separable metric spaces, we also extend them to arbitrary gauge families. It was already explicit in Hausdorff's original paper~\cite{Haus19} that his dimension could be defined via various ``lenses'' that we now call {\it gauge functions}. In fact, one often uses, as we do here, a {\it gauge family} $\varphi$, which is a one-parameter family of gauge functions $\varphi_s$ for $s \in (0,\infty)$. For each separable metric space $X$, each gauge family $\varphi$, and each set $E \subseteq X$, the classical \emph{$\varphi$-gauged Hausdorff dimension} $\dimh^\varphi(E)$ and \emph{$\varphi$-gauged packing dimension} $\dimpack^\varphi(E)$ are thus well-defined. In this paper, for each separable metric space $X$, each gauge family $\varphi$, and each point $x \in X$, we define the \emph{$\varphi$-gauged algorithmic dimension} $\dim^\varphi(x)$ and the \emph{$\varphi$-gauged strong algorithmic dimension} $\Dim^\varphi(x)$ of the point $x$. We should mention here that there is a particular gauge family $\theta$ that gives the ``un-gauged'' dimensions in the sense that the identities $\dimh^\theta(E) = \dimh(E)$, $\dimpack^\theta(E) = \dimpack(E)$, $\dim^\theta(x)= \dim(x)$, and $\Dim^\theta(x) = \Dim(x)$
always hold.

Our first two main results (Theorems~\ref{theo31} and~\ref{the32}) extend the point-to-set principle to arbitrary separable metric spaces and a wide variety of gauge families, proving that, for every separable metric space $X$, every gauge family $\varphi$ satisfying mild asymptotic constraints, and every set $E \subseteq X$,
\begin{equation}\label{eq:gp2sh}
	\dimh^{\varphi}(E)=\adjustlimits\min_{A\subseteq\N}\sup_{x\in E}\dim^{{\varphi},A}(x)
\end{equation}
and
\begin{equation}\label{eq:gp2sp}
	\dimpack^{\varphi}(E)=\adjustlimits\min_{A\subseteq\N}\sup_{x\in E}\Dim^{{\varphi},A}(x).
\end{equation}
Various nontrivial modifications to both machinery and proofs are necessary in getting from~\eqref{eq:p2sh} and~\eqref{eq:p2sp} to~\eqref{eq:gp2sh} and~\eqref{eq:gp2sp}. 

As an illustration of the power of our approach, we investigate the dimensions of hyperspaces. The \emph{hyperspace} $\cK(X)$ of a metric space $X$ is the set of all nonempty compact subsets of $X$, equipped with the Hausdorff metric~\cite{oWill04}. (For example, the ``stages'' $E_0,E_1,E_2,\ldots$ of a self-similar fractal $E\subseteq\R^n$ converge to $E$ in the hyperspace $\R^n$.) The hyperspace of a separable metric space is itself a separable metric space, and the hyperspace is typically infinite-dimensional, even when the underlying metric space is finite-dimensional. One use of gauge families is reducing such infinite dimensions to enable quantitative comparisons. For example, McClure~\cite{McClur95} defined, for each gauge family $\varphi$, a \emph{jump} $\widetilde{\varphi}$ (our notation) that is also a gauge family, and he proved~\cite{McC96} for every \emph{self-similar} subset $E$ of a separable metric space $X$,
\[\dimh^{\widetilde{\theta}}(\cK(E))=\dimh(E),\]
where $\theta$ is the above-mentioned ``un-gauged'' gauge family.

Here we prove a \emph{hyperspace dimension theorem} for the upper and lower Minkowski (i.e., box-counting) dimensions $\dimm$ and $\Dimm$. This states that, for every separable metric space $X$, \emph{every} gauge family $\varphi$, and \emph{every} $E\subseteq X$,
\begin{equation}\label{eq:hdtintro1}
	\dimm^{\widetilde{\varphi}}(\cK(E))=\dimm^\varphi(E)
\end{equation}
and 
\begin{equation} \label{eq:hdtintro2}
	\Dimm^{\widetilde{\varphi}}(\cK(E))=\Dimm^\varphi(E).
\end{equation}
We note that it is implicit in~\cite{McClur95} that these identities hold for totally bounded sets $E$ and gauge families $\varphi$ satisfying a doubling condition.

Our third main result (Theorem~\ref{thm:packminkequiv}) says that, for every separable metric space $X$, every ``well-behaved'' gauge family $\varphi$, and every compact set $E\subseteq X$,
\begin{equation}\label{eq:hdtintro3}
	\dimpack^{\widetilde{\varphi}}(\cK(E))=\Dimm^{\widetilde{\varphi}}(\cK(E)).
\end{equation}
Our proof of this result makes essential use of~\eqref{eq:hdtintro2} and the point-to-set principle~\eqref{eq:gp2sp}.

Finally, we use the point-to-set principle~\eqref{eq:gp2sp}, the identities~\eqref{eq:hdtintro2} and~\eqref{eq:hdtintro3}, and some additional machinery to prove the \emph{hyperspace packing dimension theorem} (Theorem~\ref{thm:hpdt}), which says that, for every separable metric space $X$, every well-behaved gauge family $\varphi$, and every \emph{analytic} (i.e., $\mathbf{\Sigma}^1_1$, an analog of NP that Sipser famously investigated~\cite{Sips83,Sips84,Sips92}) set $E\subseteq X$,
\begin{equation}\label{eq:hdtintro4}
	\dimpack^{\widetilde{\varphi}}(\cK(E))\geq \dimpack^\varphi(E).
\end{equation}
It is implicit in~\cite{McClur95} that~\eqref{eq:hdtintro4} holds for all $\sigma$-compact sets $E$. 

At the time of this writing it is an open question whether there is an analogous hyperspace dimension theorem for Hausdorff dimension.

David Hilbert famously wrote the following~\cite{Hilb25}.
\begin{quote}
	The final test of every new theory is its success in answering preexistent questions that the theory was not specifically created to answer.
\end{quote}
The theory of algorithmic dimensions passed Hilbert's final test when the point-to-set principle gave us the the results in the first paragraph of this introduction. We hope that the machinery developed here will lead to further such successes in the wider arena of separable metric spaces.

\section{Gauged Classical Dimensions}\label{sec:prelim}
We review the definitions of gauged Hausdorff, packing, and Minkowski dimensions. We refer the reader to~\cite{Falc14,Matt95} for a complete introduction and motivation.

	Let $(X,\rho)$ be a metric space where $\rho$ is the metric. (From now on we will omit $\rho$ when referring to the space $(X,\rho)$.) $X$ is \emph{separable} if there exists a countable set $D\subseteq X$ that is \emph{dense} in $X$, meaning that for every $x\in X$ and $\delta>0$, there is a $d\in D$ such that $\rho(x,d)<\delta$. The \emph{diameter} of a set $E\subseteq X$ is $\diam(E)=\sup\myset{\rho(x,y)}{x,y\in E}$; notice that the diameter of a set can be infinite. A \emph{cover} of $E\subseteq X$ is a collection $\mathcal{U}\subseteq\mathcal{P}(X)$ such that $E\subseteq\bigcup_{U\in\mathcal{U}}U$, and a \emph{$\delta$-cover} of $E$ is a cover $\mathcal{U}$ of $E$ such that $\diam(U)\leq\delta$ for all $U\in\mathcal{U}$.
	
	\begin{defn}[gauge functions and families]
		A \emph{gauge function} is a continuous,\footnote{Some authors require only that the function is right-continuous when working with Hausdorff dimension and left-continuous when working with packing dimension. Indeed, left continuity is sufficient for our hyperspace packing dimension theorem.} nondecreasing function from $[0, \infty)$ to $[0, \infty)$ that vanishes only at 0~\cite{Haus19, Roge98}. A \emph{gauge family} is a one-parameter family $\varphi=\myset{\varphi_s}{s\in(0, \infty)}$ of gauge functions $\varphi_s$ satisfying \[\varphi_{s}(\delta)=o(\varphi_{t}(\delta)) \mbox{ as }\delta\to 0^{+}\] whenever $s>t$.
	\end{defn}

	The \emph{canonical gauge family} is $\theta=\myset{\theta_s}{s\in(0,\infty)}$, defined by $\theta_s(\delta)=\delta^s$. ``Un-gauged'' or ``ordinary'' Hausdorff, packing, and Minkowski dimensions are special cases of the following definitions, using $\varphi=\theta$.

	Some of our gauged dimension results will require the existence of a ``precision family'' for the gauge family.
	\begin{defn}[precision family]
		A \emph{precision sequence} for a gauge function $\varphi$ is a function $\alpha:\N\to\Q^+$ that vanishes as $r\to\infty$ and satisfies $\varphi(\alpha(r))=O(\varphi(\alpha(r+1)))$ as $r\to\infty$. A \emph{precision family} for a gauge family $\varphi=\{\varphi_s\mid s\in(0,\infty)\}$ is a one-parameter family $\alpha=\{\alpha_s\mid s\in(0,\infty)\}$ of precision sequences satisfying
		\[\sum_{r\in\N}\frac{\varphi_t(\alpha_s(r))}{\varphi_s(\alpha_s(r))}<\infty\]
		whenever $s<t$.
	\end{defn}
	\begin{observation}
		$\alpha_s(r)=2^{-sr}$ is a precision family for the canonical gauge family $\theta$.
	\end{observation}

	\begin{defn}[gauged Hausdorff measure and dimension]
		For every metric space $X$, set $E\subseteq X$, and gauge function $\varphi$, the \emph{$\varphi$-gauged Hausdorff measure} of $E$ is
		\[H^\varphi(E)=\lim_{\delta\to 0^+}\inf\left\{\sum_{U\in\mathcal{U}}\varphi(\diam(U))\;\middle|\;\mathcal{U}\text{ is a countable $\delta$-cover of }E\right\}.\]
		 For every gauge family $\varphi=\myset{\varphi_s}{s\in(0, \infty)}$, the \emph{$\varphi$-gauged Hausdorff dimension} of $E$ is
		\[\dimh^\varphi(E)=\inf\myset{s\in(0,\infty)}{H^{\varphi_s}(E)=0}.\]
	\end{defn}

	\begin{defn}[gauged packing measure and dimension]
		For every metric space $X$, set $E\subseteq X$, and $\delta\in(0,\infty)$, let $\mathcal{V}_\delta(E)$ be the set of all countable collections of disjoint open balls with centers in $E$ and diameters at most $\delta$. For every gauge function $\varphi$ and $\delta>0$, define the quantity
		\[P^\varphi_\delta(E)=\sup_{\mathcal{U}\in\mathcal{V}_\delta(E)}\sum_{U\in\mathcal{U}}\varphi(\diam(U)).\]
		Then the \emph{$\varphi$-gauged packing pre-measure} of $E$ is
		\[P^\varphi_0(E)=\lim_{\delta\to 0^+}P^\varphi_\delta(E),\]
		and the \emph{$\varphi$-gauged packing measure} of $E$ is
		\[P^\varphi(E)=\inf\left\{\sum_{U\in\mathcal{U}} P^\varphi_0(U)\;\middle|\;\mathcal{U}\text{ is a countable cover of }E\right\}.\]
		For every gauge family $\varphi=\myset{\varphi_s}{s\in(0, \infty)}$, the \emph{$\varphi$-gauged packing dimension} of $E$ is
		\[\dimpack^\varphi(E)=\inf\myset{s\in(0,\infty)}{P^{\varphi_s}(E)=0}.\]
	\end{defn}

	\begin{defn}[gauged Minkowski dimensions]
		For every metric space $X$, $E\subseteq X$, and $\delta\in(0,\infty)$, let
		\[N(E,\delta)=\min\myset{|F|}{F\subseteq X\text{ and }E\subseteq \bigcup_{x\in F}B_\delta(x)},\]
		where $B_\delta(x)$ is the open ball of radius $\delta$ centered at $x$. Then for every gauge family $\varphi=\{\varphi_s\}_{s\in(0, \infty)}$ the \emph{$\varphi$-gauged lower} and \emph{upper Minkowski dimension} of $E$ are
		\[\dimm^\varphi(E)=\inf\left\{s\;\middle|\;\liminf_{\delta\to 0^+} N(E,\delta)\varphi_s(\delta)=0\right\}\]
		and
		\[\Dimm^\varphi(E)=\inf\left\{s\;\middle|\;\limsup_{\delta\to 0^+} N(E,\delta)\varphi_s(\delta)=0\right\},\]
		respectively.
	\end{defn}

	When $X$ is separable, it is sometimes useful to require that the balls covering $E$ have centers in the countable dense set $D$. For all $E\subseteq X$ and $\delta\in(0,\infty)$, let \[\hat{N}(E,\delta)=\min\myset{|F|}{F\subseteq D\text{ and }E\subseteq \bigcup_{x\in F}B_\delta(x)}.\]
	\begin{observation}\label{obs:separabledim}
		If $X$ is a separable metric space and $\varphi=\{\varphi_s\}_{s\in(0, \infty)}$ is a gauge family, then for all $E\subseteq X$,
		\begin{enumerate}

		\vspace{0.5em}
		\item $\displaystyle\dimm^\varphi(E)=\inf\left\{s\;\middle|\;\liminf_{\delta\to 0^+} \hat{N}(E,\delta)\varphi_s(\delta)=0\right\}$.

		\vspace{0.5em}
		\item $\displaystyle\Dimm^\varphi(E)=\inf\left\{s\;\middle|\;\limsup_{\delta\to 0^+} \hat{N}(E,\delta)\varphi_s(\delta)=0\right\}$.
		\end{enumerate}
	\end{observation}

	The following relationship between upper Minkowski dimension and packing dimension was previously known to hold for the canonical gauge family $\theta$, a result that is essentially due to Tricot~\cite{Tric82}. Our proof of this gauged generalization, which is in the appendix, is adapted from the presentation by Bishop and Peres~\cite{BisPer16} of the un-gauged proof.

	\begin{lemma}[generalizing Tricot~\cite{Tric82}]\label{lem:tricot}
	Let $X$ be any metric space, $E\subseteq X$, and $\varphi$ a gauge family.
	\begin{enumerate}
		\item If $\varphi_t(2\delta)=O(\varphi_s(\delta))$ as $\delta\to 0^+$ for all $s<t$, then \[\displaystyle\dimpack^\varphi(E) \geq \inf\left\{\sup_{i\in\N} \Dimm^\varphi(E_i)\,\middle|\, E\subseteq \bigcup_{i\in\N} E_i\right\}.\]
		\item If there is a precision family for $\varphi$, then \[\displaystyle\dimpack^\varphi(E) \leq \inf\left\{\sup_{i\in\N} \Dimm^\varphi(E_i)\,\middle|\, E\subseteq \bigcup_{i\in\N} E_i\right\}.\]
	\end{enumerate}
	\end{lemma}

\section{Gauged Algorithmic Dimensions}\label{sec:gad}

In this section we formulate algorithmic dimensions in arbitrary
separable metric spaces and with arbitrary gauge families.

For the rest of this paper, let $X=(X,\rho)$ be a separable metric space, and fix a function $f:\binary\to X$ such that the set $D=\range(f)$ is dense in $X$. The metric space $X$ is \emph{computable} if there is a computable function $g:(\binary)^2\times \Q^+ \to \Q$ that approximates $\rho$ on $D$ in the sense that, for all $v, w\in\binary$ and $\delta\in\Q^+$. \[|g(v, w, \delta)-\rho(f(v), f(w))|\le \delta.\] Our results here hold for all separable metric spaces, whether or not they are computable, but our methods make explicit use of the function $f$.

Following standard practice \cite{Nies09, DowHir10, LiVit19}, fix a
universal oracle Turing machine $U$, and define the \emph{(plain)
Kolmogorov complexity\/} of a string $w\in \binary$ \emph{relative
to} an oracle $A\subseteq \N$ to be
\[\C^A(w)=\min\myset{|\pi|}{\pi\in\binary\mbox{ and }U^A(\pi)=w},\]
i.e., the minimum number of bits required to cause $U$ to output $w$
when it has access to the oracle $A$. The \emph{(plain) Kolmogorov complexity} of $w$ is then $\C(w)=\C^\emptyset(w)$.

We define the \emph{(plain) Kolmogorov complexity} of a point $q\in D$ to
be
\[\C(q)=\min\myset{\C(w)}{w\in\binary\mbox{ and }f(w)=q},\]
noting that this depends on the enumeration $f$ of $D$ that we have
fixed.

The \emph{Kolmogorov complexity\/} of a point $x\in X$ at \emph{precision} $\delta\in (0,\infty)$ is
\[\C_\delta(x)=\min\myset{\C(q)}{q\in D\mbox{ and }\rho(q,x)< \delta}.\]
The {\em algorithmic dimension\/} of a point $x\in X$ is
\begin{equation}\label{eq21}
	\dim(x) = \liminf_{\delta\to 0^+} \frac{\C_\delta(x)}{\log(1/\delta)},
\end{equation}
and the {\em strong algorithmic dimension\/} of $x$ is
\begin{equation}\label{eq22}
	\Dim(x) = \limsup_{\delta\to 0^+} \frac{\C_\delta(x)}{\log(1/\delta)}.
\end{equation}

These two dimensions\footnote{The definitions given here differ slightly from the standard formulation in which prefix Kolmogorov complexity is used instead of plain Kolmogorov complexity and the precision parameter $\delta$ belongs to $\{2^{-r}\mid r\in\N\}$. The present formulation is equivalent to the standard one for un-gaugued dimensions and facilitates our generalization to gauged algorithmic dimensions. In particular, plain Kolmogorov complexity is only needed to accommodate gauge functions $\varphi$ in which the convergence of $\varphi$ to 0 as $\delta\to 0^+$ is very slow.} have been extensively investigated in the special cases where $X$ is a Euclidean space $\Rn$ or a sequence space $\Sigma^{\omega}$~\cite{LutLut20,DowHir10}.

Having generalized algorithmic dimensions to arbitrary separable
metric spaces, we now generalize them to arbitrary gauge families.

Let $\varphi=\myset{\varphi_s}{s\in(0, \infty)}$ be a gauge family. Then, the {\em $\varphi$-gauged algorithmic dimension} of a point $x\in X$ is
\begin{equation}\label{eq23}
	\dim^{\varphi}(x) =\inf\myset{s}{\liminf_{\delta\to 0^+}2^{\C_\delta(x)}\varphi_s(\delta)=0},
\end{equation}
and the {\em $\varphi$-gauged strong algorithmic dimension\/} of $x$ is
\begin{equation}\label{eq24}
	\Dim^{\varphi}(x) =\inf\myset{s}{\limsup_{\delta\to 0^+}2^{\C_\delta(x)}\varphi_s(\delta)=0},
\end{equation}

Gauged algorithmic dimensions $\dim^{\varphi}(x)$ have been
investigated by Staiger \cite{Stai17}\ in the special case where $X$
is a sequence space $\Sigma^{\omega}$.

A routine inspection of~\eqref{eq21}--\eqref{eq24} verifies the following.
\begin{observation}\label{obs21} For all $x\in X$, $\dim^{\theta}(x)=\dim(x)$ and $\Dim^{\theta}(x)=\Dim(x)$,
where $\theta$ is the canonical gauge family given
by $\theta_s(\delta)=\delta^s$.
\end{observation}

A specific investigation of algorithmic (or classical) dimensions
might call for a particular gauge function on family for one of two reasons.
First, many gauge functions may assign the same dimension to an
object under consideration (because they converge
to 0 at somewhat similar rates as $\delta\to 0^{+}$) but additional
considerations may identify one of these as being the
most precisely tuned to the phenomenon of interest. Finding such a
gauge function is called finding the ``exact dimension'' of the object
under investigation. This sort of calibration has been studied
extensively for classical dimensions \cite{Falc14,Roge98}\ and by
Staiger \cite{Stai17}\ for algorithmic dimension.

The second reason, and the reason of interest to us here, why
specific investigations might call for particular gauge families is
that a given gauge family $\varphi$ may be so completely out of tune
with the phenomenon under investigation that the $\varphi$-gauged
dimensions of the objects of interest are either all minimum (all 0)
or else all maximum (all the same dimension as the space $X$
itself). In such a circumstance, a gauge family that converges to 0
more quickly or slowly than $\varphi$ may yield more informative
dimensions. Several such circumstances were investigated in a
complexity-theoretic setting by Hitchcock, J. Lutz, and Mayordomo
\cite{SDNC}.

The following routine observation indicates the direction in which
one adjusts a gauge family's convergence to 0 in order to adjust the
resulting gauged dimensions upward or downward.

\begin{observation}\label{obs22}
	If $\varphi$ and $\psi$ are gauge families with $\varphi_s(\delta)=o(\psi_s(\delta))$ as $\delta\to 0^{+}$ for all $s\in(0,\infty)$, then, for all $x\in X$, $\dim^{\varphi}(x)\le \dim^{\psi}(x)$ and $\Dim^{\varphi}(x)\le \Dim^{\psi}(x)$.
\end{observation}

We now define an operation on gauge families that is implicit in earlier work~\cite{McC96} and is explicitly used in the results of Section~\ref{sec:hdts}.

\begin{defn}[jump]
	The {\em jump\/} of a gauge family $\varphi$ is the family $\widetilde{\varphi}$ given $\widetilde{\varphi}_s(\delta)=2^{-1/\varphi_s(\delta)}$.
\end{defn}

\begin{observation}\label{obs23}
	The jump of a gauge family is a gauge family.
\end{observation}

We now note that the jump of a gauge family always converges to 0
more quickly than the original gauge family.

\begin{lemma}\label{lemm24}
	For all gauge families $\varphi$ and all $s\in (0, \infty)$, $\widetilde{\varphi}_s(\delta)=o(\varphi_s(\delta))$ as $\delta\to 0^{+}$.
\end{lemma}

Observation \ref{obs23}\ and Lemma \ref{lemm24}\ immediately imply
the following.

\begin{corollary}\label{cor25} For all gauge families $\varphi$ and all $x\in
X$, $\dim^{\widetilde{\varphi}}(x)\le \dim^{\varphi}(x)$ and $\Dim^{\widetilde{\varphi}}(x)\le \Dim^{\varphi}(x)$.
\end{corollary}

The definitions and results of this section relativize to arbitrary
oracles $A\subseteq \N$ in the obvious manner, so the Kolmogorov
complexities $\C^A(q)$ and $\C^A_\delta(x)$ and the dimensions
$\dim^A(x)$, $\Dim^A(x)$, $\dim^{\varphi, A}(x)$, and
$\Dim^{\varphi, A}(x)$ are all well-defined and behave as indicated.

\begin{observation}\label{obs:36}
	For all gauge families $\varphi$, all $x\in X$, and all $s>0$,
	\[\log\big(2^{\C_\delta(x)}\widetilde{\varphi}_s(\delta)\big)=\frac{\C_\delta(x)\varphi_s(\delta)-1}{\varphi_s(\delta)}.\]
\end{observation}

The $\widetilde{\varphi}$-gauged algorithmic dimensions admit the following characterizations, the second of which is used in the proof of our hyperspace packing dimension theorem.

\begin{theorem}\label{thm:DimL}
	For all gauge families $\varphi$ and all $x\in X$, the following identities hold.
	\begin{enumerate}
		\item $\displaystyle \dim^{\widetilde{\varphi}}(x)=\inf\left\{s\;\middle|\;\liminf_{\delta\to 0^+}\C_\delta(x)\varphi_s(\delta)=0\right\}$. 
		\item $\displaystyle \Dim^{\widetilde{\varphi}}(x)=\inf\left\{s\;\middle|\;\limsup_{\delta\to 0^+}\C_\delta(x)\varphi_s(\delta)=0\right\}$.
	\end{enumerate}
\end{theorem}

\section{The General Point-to-Set Principle}\label{sec3}

We now show that the point-to-set principle of J. Lutz and N. Lutz~\cite{LutLut18}\ holds in arbitrary separable metric spaces and for gauged dimensions. The proofs of these theorems, which can be found in the appendix, are more delicate and involved than those in~\cite{LutLut18}. This is partially due to the fact that the metric spaces here need not be finite-dimensional, and to the weak restrictions we place on the gauge family.

\begin{theorem}[general point-to-set principle for Hausdorff dimension]\label{theo31}
For every separable metric space $X$, every gauge family $\varphi$, and every set $E\subseteq X$,
\[\dimh^{\varphi}(E)\geq\adjustlimits\min_{A\subseteq\N}\sup_{x\in E}\dim^{{\varphi},A}(x).\]
Equality holds if there is a precision family for $\varphi$.
\end{theorem}

\begin{theorem}[general point-to-set principle for packing dimension]\label{the32}
	Let $X$ be any separable metric space, $E\subseteq X$, and $\varphi$ a gauge family.
	\begin{enumerate}
	\item If $\varphi_t(2\delta)=O(\varphi_s(\delta))$ and $\varphi_s(\delta)=O(1/\log\log(1/\delta))$ as $\delta\to 0^+$ for all $s<t$, then \[\dimpack^{\varphi}(E)\geq\adjustlimits\min_{A\subseteq\N}\sup_{x\in E}\Dim^{{\varphi},A}(x).\]
	\item If there is a precision family for $\varphi$, then
	\[\dimpack^{\varphi}(E)\leq\adjustlimits\min_{A\subseteq\N}\sup_{x\in E}\Dim^{{\varphi},A}(x).\]
	\end{enumerate}
\end{theorem}

\begin{proof}[Proof of Theorem~\ref{the32}]
\begin{enumerate}
\item Assume that $\varphi_t(2\delta)=O\left(\varphi_s(\delta)\right)$ and $\varphi_s(\delta)=O(1/\log\log(1/\delta))$ hold for all $s<t$. It suffices to show that there exists $A\subseteq\N$ such that, for all $x\in E$,
\begin{equation}\label{eq:the32A}
	\Dim^{\varphi,A}(x)\leq \dimpack^\varphi(E).
\end{equation}

Let $s>t>u>\dimpack^\varphi(E)$. Since $u>\dimpack^\varphi(E)$, Lemma~\ref{lem:tricot} and our hypothesis on $\varphi$ tell us that there is a cover $\{E_i\}_{i\in\Z^+}$ of $E$ such that, for all $i\in\Z^+$,
\begin{equation}\label{eq:the32B}
	\Dimm^\varphi(E_i)\leq u.
\end{equation}

For each $i\in\Z^+$ and $\delta\in\Q\cap(0,1)$, let $F(i,\delta)\subseteq D$ satisfy \[|F(i,\delta)|=\hat{N}(E_i,\delta)\]
and
\[E_i\subseteq\bigcup_{d\in F(i,\delta)} B_\delta(d).\]
Define $h:\Z^+\times\Q\cap(0,1)\to(\{0,1\}^*)^{\hat{N}(E_i,\delta)}$ by
\[h(i,\delta)=\big(w_{i,\delta,1},\ldots,w_{i,\delta,\hat{N}(E_i,\delta)}\big),\]
where, recalling that $f$ is the function mapping bit strings onto the dense set $D$,
\[F(i,\delta)=\left(f\big(w_{i,\delta,1}\big),\ldots,f\big(w_{i,\delta,\hat{N}(E_i,\delta)}\big)\right).\]
Let $A$ be an oracle encoding $h$.

To prove~\eqref{eq:the32A}, let $x\in E$. It suffices to show that
\[\lim_{\delta\to 0^+} 2^{\C^A_\delta(x)}\varphi_s(\delta)=0.\]
For this, let $\varepsilon>0$. It suffices to show that, for all sufficiently small $\delta\in\Q^+$,
\begin{equation}\label{eq:the32C}
	\C^A_\delta(x)<\log\frac{\varepsilon}{\varphi_s(\delta)}.
\end{equation}

For each $\delta\in\Q\cap(0,1)$, let $r(\delta)=\left\lceil\log\frac{1}{\delta}\right\rceil$ and $\delta'=2^{-r(\delta)}$, so that $\frac{\delta}{2}<\delta'\leq\delta$. Since $s>t$, our hypothesis on $\varphi$ tells us that there is a constant $a>0$ such that, for all sufficiently small $\delta\in\Q^+$,
\begin{equation}\label{eq:the32D}
	\frac{1}{\varphi_t(\delta')}\leq\frac{a}{\varphi_s(2\delta')}\leq\frac{a}{\varphi_s(\delta)}.
\end{equation}
Since $t>u$,~\eqref{eq:the32B} tells us that, for all $i\in\N$,
\[\lim_{\delta\to 0^+}\hat{N}(E_i,\delta)\varphi_t(\delta)=0.\]
Hence, for all $i\in\N$ and all sufficiently small $\delta\in\Q^+$,
\begin{equation}\label{eq:the32E}
	\hat{N}(E_i,\delta)\varphi_t(\delta)<\frac{\varepsilon}{2a}.
\end{equation}
In particular, then,~\eqref{eq:the32D} and~\eqref{eq:the32E} tell us that, for all sufficiently small $\delta\in\Q^+$,
\begin{equation}\label{eq:the32F}
	\hat{N}(E_i,\delta')\leq\frac{\varepsilon}{2a\varphi_t(\delta')}\leq\frac{\varepsilon}{2\varphi_s(\delta)}.
\end{equation}

For each $i,k\in\Z^+$ and $\delta\in\Q\cap (0,1)$, let $\pi\in\{0,1\}^*$ be a string that encodes $i$, $r(\delta)$, and $k$, with
\[|\pi|=\log k + O(\log i +\log r(\delta)).\]

Let $M$ be an oracle Turing machine that, with oracle $A$ and program $\pi$, outputs the string $w_{i,\delta',k}$ that is the $k$\textsuperscript{th} component of $h(i,\delta')$ (if there is one), where $\delta'=2^{-r(\delta)}$. Let $c_M$ be an optimality constant for $M$.

To see that~\eqref{eq:the32C} holds, choose $i\in\Z^+$ such that $x\in E_i$, and let $\delta\in\Q\cap(0,1)$. Let $\delta'=2^{-r(\delta)}$, and choose $k\in\big\{1,\ldots,\hat{N}(E_i,\delta')\big\}$ such that $x\in B_{\delta'}\left(f\big(w_{i,\delta',k}\big)\right)$. Then
\[f\big(w_{i,\delta',k}\big)\in D\cap B_{\delta'}(x)\subseteq D\cap B_\delta(x),\]
so~\eqref{eq:the32F} gives, for all sufficiently small $\delta\in\Q^+$,
\begin{align*}
	\C^A_\delta(x)
	&\leq\C^A\big(w_{i,\delta',k}\big)\\
	&\leq\C_M^A\big(w_{i,\delta',k}\big)+c_M\\
	&\leq |\pi|+c_M\\
	&\leq \log k + c_M + O(\log i+\log r(\delta))\\
	&\leq \log\hat{N}(E_i,\delta')+O(\log i+\log r(\delta))\\
	&\leq \log\frac{\varepsilon}{2\varphi_s(\delta)}+O(\log i+\log r(\delta)).
\end{align*}
Since $i$ is a constant and, by our assumption, $\log r(\delta)\leq \log(\log(1/\delta)+1)=O(1/\varphi_t(\delta))=o(1/\varphi_s(\delta))$,
the second term vanishes as $\delta\to 0^+$, affirming~\eqref{eq:the32C}.

\item Let $s>t>\sup_{x\in E}\Dim^{\varphi,A}(x)<t<s$. Then for each $x\in E$ and all sufficiently small $\delta\in\Q^+$, $\C_\delta^{A}(x)< \log(1/\varphi_{t}(\delta))$.
For all $\delta\in\Q^+$, let
\[\mathcal{U}_\delta=\left\{B_{\delta}(f(w))\;\middle|\;\C^{A}(w)\le \log (1/\varphi_{t}(\delta))\right\},\]
and for each $i\in\N$, let
\[E_i=\left\{x\;\middle|\;\forall \delta<1/i,\ x\in \mathcal{U}_\delta\right\}.\]
Then $E\subseteq\bigcup_{i\in\N} E_i$.
For each $\delta<1/i$, $N(E_i, \delta)<2/\varphi_{t}(\delta)$,
so $N(E_i, \delta)\varphi_{s}(\delta)=o(1)$, and therefore $\Dimm^\varphi(E_i)\leq s$. Assuming that there is a precision family for $\varphi$, the result follows by Lemma~\ref{lem:tricot}.
\end{enumerate}
\end{proof}

\section{Hyperspace Dimension Theorems}\label{sec:hdts}
This section presents our main theorems.

As before, let $X=(X,\rho)$ be a separable metric space. The \emph{hyperspace} of $X$ is the metric space $\cK(X)= (\cK(X), \rhoH)$,
where $\cK(X)$ is the set of all nonempty compact subsets of $X$ and
$\rhoH$ is the {\em Hausdorff metric\/} \cite{Haus14}\ on $\cK(X)$
defined by
\[\rhoH(E,F) = \max\left\{\sup_{x\in E} \rho(x,F), \sup_{y\in F} \rho(E,y)\right\},\]
where $\rho(x, F)=\inf_{y\in F} \rho(x,y)$ and $\rho(E, y)=\inf_{x\in E}
\rho(x,y)$.

Let $f:\binary\to X$ and $D=\range(f)$ be fixed as at the
beginning of section \ref{sec:gad}, so that $D$ is dense in $X$. Let
$\cD$ be the set of all nonempty, finite subsets of $D$. It is well
known and easy to show that $\cD$ is a countable dense subset of
$\cK(X)$, and it is routine to define from $f$ a function
$\widetilde{f}: \binary \to \cK(X)$ such that $\range(\widetilde{f})=\cD$.
Thus $\cK(X)$ is a separable metric space, and the results in
section \ref{sec3}\ hold for $\cK(X)$.

It is important to note the distinction between the classical
Hausdorff and packing dimensions $\dimh(E)$ and $\dimpack(E)$ of a nonempty compact subset $E$ of
$X$ and the algorithmic dimensions $\dim(E)$ and $\Dim(E)$ of this same set when it
is regarded as a point in $\cK(X)$. In the appendix, we construct an example of a set $E$ with Hausdorff and packing dimensions $\dimh(E)=\dimpack(E)=\log(2)/\log(4)\approx 0.356$ and $\dim(E)=\Dim(E)=\infty$.

Our first hyperspace dimension theorem applies to lower and upper Minkowski dimensions. This theorem, which is proven using a counting argument, is very general, placing no restrictions on the gauge family $\varphi$ or the separable metric space $X$.

\begin{theorem}[hyperspace Minkowski dimension theorem]\label{thm:hypmink}
	For every gauge family $\varphi$ and every $E\subseteq X$,
	\[\dimm^{\widetilde{\varphi}}(\cK(E))= \dimm^\varphi(E)\quad\text{and}\quad\Dimm^{\widetilde{\varphi}}(\cK(E))= \Dimm^\varphi(E).\]
\end{theorem}

Our third main result is the surprising fact that in a hyperspace, packing dimension and upper Minkowski dimension are equivalent for compact sets.

\begin{theorem}\label{thm:packminkequiv}
	For every separable metric space $X$, every compact set $E\subseteq X$, and every gauge family $\varphi$ such that $\varphi_t(2\delta)=O(\varphi_s(\delta))$ and $\varphi_s(\delta)=O(1/\log\log(1/\delta))$ as $\delta\to 0^+$ for all $s<t$ and there is a precision family for $\varphi$,
	\[\dimpack^{\widetilde{\varphi}}(\cK(E))=\Dimm^{\widetilde{\varphi}}(\cK(E)).\]
\end{theorem}

The point-to-set principle is central to our proof of this theorem: We recursively construct a single compact set $L\subseteq E$ (i.e., a single point in the hyperspace $\cK(E)$) that has high Kolmogorov complexity at infinitely many precisions, relative to an appropriate oracle $A$. We then invoke Theorem~\ref{thm:DimL} to show that this $L$ has high $\varphi$-gauged strong algorithmic dimension relative to $A$. By the point-to-set principle, then, $\cK(E)$ has high packing dimension.

\begin{observation}
	The conclusion of Theorem~\ref{thm:packminkequiv} does not hold for arbitrary sets $E$.
\end{observation}
\begin{proof}
	Let $E=\{1/n:n\in\N\}$. Then $\Dimm^\theta(E)=1/2$, but every compact subset of $E$ is finite, so $\cK(E)$ is countable and $\dimpack^{\widetilde{\theta}}(\cK(E))=0$.
\end{proof}

\begin{theorem}[hyperspace packing dimension theorem]\label{thm:hpdt}
If $X$ is a separable metric space, $E\subseteq X$ is an analytic set, and $\varphi$ is a gauge family such that $\varphi_s(2\delta)=O(\varphi_s(\delta))$ and $\varphi_s(\delta)=O(1/\log\log(1/\delta))$ as $\delta\to 0^+$ for all $s\in(0,\infty)$ and there is a precision family for $\varphi$, then
\[\dimpack^{\widetilde{\varphi}}(\cK(E))\geq\dimpack^\varphi(E).\]
\end{theorem}
\begin{proof}
	For compact sets $E$, Theorem~\ref{thm:packminkequiv} and the hyperspace Minkowski dimension theorem (Theorem~\ref{thm:hypmink}) imply $\dimpack^{\widetilde{\varphi}}(\cK(E))=\Dimm^\varphi(E)$.

	A result of Joyce and Preiss (Corollary 1 in~\cite{JoyPre95}) states that every analytic set with positive (possibly infinite) gauged packing measure contains a compact subset with positive (finite) packing measure in the same gauge. It follows that if $E$ is analytic, then for all $\varepsilon>0$ there exists a compact subset $E_\varepsilon\subseteq E$ with $\dimpack^\varphi(E_\varepsilon)\geq \dimpack^\varphi(E)-\varepsilon$.
	Therefore
	\begin{align*}
		\dimpack^{\widetilde{\varphi}}(\cK(E_\varepsilon))&=\Dimm^\varphi(E_\varepsilon)\\
		&\geq \dimpack^\varphi(E_\varepsilon)\\
		&\geq \dimpack^\varphi(E)-\varepsilon.
	\end{align*}
	Letting $\varepsilon\to 0$ completes the proof.
\end{proof}

\section{Conclusion}

Our results exhibit and amplify the power of the theory of computing to make unexpected contributions to other areas of the mathematical sciences. We hope and expect to see more such results in the near future.

We mention three open problems whose solutions may contribute to such progress. First, at the time of this writing, a hyperspace Hausdorff dimension theorem remains an open problem. The difficulty in adapting our approach to that problem is that in the proof of Theorem~\ref{thm:packminkequiv}, the set $L$ we construct is only guaranteed to have high complexity at infinitely many precisions. An analogous proof for Hausdorff dimension would require constructing a set $L$ that has high complexity at all but finitely many precisions.

Second, it would be useful to identify classes of spaces in which Billingsley-type algorithmic dimensions---dimensions shaped by probability measures---can be formulated.

Finally, we do not at this time know how to characterize algorithmic dimensions in separable metric spaces in terms of martingales or more general gales. This is despite the fact that algorithmic dimensions were first formulated in these terms in sequence spaces.

\subsection*{Acknowledgments}

We thank anonymous reviewers of an early draft of this paper for several observations that have improved this paper.

\newpage

\bibliography{hdt}

\newpage

\appendix

\section{Proofs from Section 2}
\begin{observation}\label{obs:separableN}
	If $E$ is separable and $0<\delta<\hat{\delta}$, then $\hat{N}(E,\hat{\delta})\leq N(E,\delta)$.
\end{observation}
\begin{proof}
	Let $F\subseteq X$ be a witness to $N(E,\delta)$. For all $x\in F$, there exists $\hat{x}\in D$ with $0<\rho(x,\hat{x})<\hat{\delta}-\delta$, so $B_{\delta}(x)\subseteq B_{\hat{\delta}}(\hat{x})$. Thus, the set $\hat{F}=\myset{\hat{x}}{x\in F}\subseteq D$ satisfies
	\[E\subseteq\bigcup_{x\in F}B_{\delta}(x)\subseteq\bigcup_{\hat{x}\in \hat{F}}B_{\hat{\delta}}(\hat{x}),\]
	so $\hat{N}(E,\hat{\delta})\leq|\hat{F}|=|F|=N(E,\delta)$.
\end{proof}

\begin{proof}[Proof of Observation~\ref{obs:separabledim}]
	Since $\hat{N}(E,\delta)\geq N(E,\delta)$, it is clear that the left-hand side of each equation is bounded above by the right-hand side. We now prove the other direction. Fix $s\in(0,\infty)$.
	\begin{enumerate}
		\item Assume that $\liminf_{\delta\to 0^+} N(E,\delta)\varphi_s(\delta)=0$, and let		$\delta_0,\varepsilon>0$. Then there is some $\delta<\delta_0$ such that $N(E,\delta)\varphi_s(\delta)<\varepsilon/2$. By the (right) continuity of $\varphi_s$, there exists some $\hat{\delta}\in (\delta,\delta_0)$ such that $\varphi_s(\hat{\delta})<2\varphi_s(\delta)$. By Observation~\ref{obs:separableN}, $\hat{N}(E,\hat{\delta})\leq N(E,\delta)$, so we have $\hat{N}(E,\hat{\delta})\varphi_s(\hat{\delta})<2N(E,\delta)\varphi_s(\delta)<\varepsilon$.
		\item Assume that $\limsup_{\delta\to 0^+} N(E,\delta)\varphi_s(\delta)=0$, and let $\varepsilon>0$. Then there is some $\delta_0>0$ such that $N(E,\delta)\varphi_s(\delta)<\varepsilon/2$ for all $\delta<\delta_0$. For every $\hat{\delta}<\delta_0$, by the (left) continuity of $\varphi_s$, there is some $\delta<\hat{\delta}$ such that $\varphi_s(\hat{\delta})<2\varphi_s(\delta)$. Observation~\ref{obs:separabledim} tells us that $\hat{N}(E,\hat{\delta})\leq N(E,\delta)$, so we have $\hat{N}(E,\hat{\delta})\varphi_s(\hat{\delta})<2N(E,\delta)\varphi_s(\delta)<\varepsilon$.
	\end{enumerate}
\end{proof}

\begin{proof}[Proof of Lemma~\ref{lem:tricot}]
We follow the presentation by Bishop and Peres~\cite{BisPer16} of the proof for the canonical gauge family.
\begin{enumerate}
\item Fix $t>s>0$, and assume there is some constant $c>0$ such that $\varphi_t(2\delta)<c\cdot\varphi_s(\delta)$ for all sufficiently small $\delta$. Let $F\subseteq E$ be a set with $P_0^{\varphi_s}(F)<\infty$. It suffices to show that
$\limsup_{\delta\to 0^+} N(F,\delta)\varphi_t(\delta)<\infty$. To see this, fix $\delta>0$ and let $N_p(F,\delta)$ denote the maximum number of disjoint open balls of diameter $\delta$ with centers in $F$, and observe that $N_p(F,\delta)\varphi_s(\delta)\leq P_\delta^{\varphi_s}(F)$ and $N(F,2\delta)\leq N_p(F,\delta)$. It follows that $N(F,2\delta)\varphi_t(2\delta)/c<\infty$ for all sufficiently small $\delta$, which yields the desired bound. 

\item Suppose that $\alpha=\{\alpha_s\}_{s\in(0,\infty)}$ is a precision family for $\varphi$, fix $t>s>0$, and suppose that $F\subseteq E$ is a set with $P_0^{\varphi_t}(F)>0$. It suffices to show that $\limsup_{\delta\to0^+}N(F,\delta)\varphi_s(\delta)>0$, and since $N_p(F,\delta)\leq N(F,\delta)$, it suffices to show that $\limsup_{\delta\to0^+}N_p(F,\delta)\varphi_s(\delta)>0$.

Let $\gamma>0$ be such that for every $\varepsilon>0$ there is a collection $\{B_{\delta_j/2}(x_j)\}_{j\in\mathbb{N}}$ of disjoint balls with diameters $\delta_j\leq\varepsilon$, centers $x_j\in F$, and $\sum_{j\in\N}\varphi_t(\delta_j)>\gamma$. Fix an $\varepsilon>0$ and a corresponding collection of balls. Let $r_0=\max\{r\in\N\mid\alpha_s(r)>\varepsilon\}$,
and for each $r\in\mathbb{N}$, let $n_r=|\{j\in\N\mid \alpha_s(r+1)\leq \delta_j<\alpha_s(r)\}|$. Then $N_p(F,\alpha_s(r+1))\geq n_r$, so
\begin{align*}
	\sum_{r=r_0}^\infty N_p(F,\alpha_s(r+1))\varphi_t(\alpha_s(r))
	&\geq\sum_{r=r_0}^\infty n_r\varphi_t(\alpha_s(r))\\
	&> \sum_{j\in\N}\varphi_t(\delta_j)\\
	&>\gamma.
\end{align*}
For every $r\in\N$, define
\[a_r=N_p(F,\alpha_s(r+1))\varphi_{s}(\alpha_s(r+1))\geq N_p(F,\alpha_s(r))\frac{\varphi_{s}(\alpha_s(r))}{c},\]
where $c$ is the implicit constant in the precision sequence condition for $\alpha_s$. Then
\begin{align*}
	\sum_{r=r_0}^\infty a_r\frac{\varphi_{t}(\alpha_s(r))}{\varphi_{s}(\alpha_s(r))}\geq \frac{1}{c}\sum_{r= r_0}^\infty N_p(F,\alpha_s(r+1)){\varphi_{t}(\alpha_s(r+1))}>\frac{\gamma}{c}.
\end{align*}

If $\{a_r\}_{r\in\N}$ were bounded, then the sum would tend to 0 as $r_0\to\infty$ since $\alpha$ is a precision family. This is a contradiction, so we have $\limsup_{r\to\infty}N_p(F,\alpha_s(r+1))\varphi_s(\alpha_s(r+1))>0$, and the claim holds.

\end{enumerate}
\end{proof}

\section{Proofs from Section 3}
\begin{proof}[Proof of Lemma~\ref{lemm24}]
Letting $x=\frac{\varphi_s(\delta)}{\ln 2}$ and noting that $x\to 0^{+}$
as $\delta\to 0^{+}$, we have
\[\widetilde{\varphi}_s(\delta)= 2^{-1/\varphi_s(\delta)}= e^{-1/x}=o(x) = o(
\varphi_s(\delta))\] as $\delta\to 0^{+}$.
\end{proof}

\begin{proof}[Proof of Theorem~\ref{thm:DimL}]
	Let $\varphi$ and $x$ be as given, and let $S^-$ and $S^+$ be the sets on the right-hand sides of 1 and 2, respectively.
	\begin{enumerate}
		\item It suffices to show that
		\[\big(\dim^{\widetilde{\varphi}},\infty\big)\subseteq S^-\subseteq \big[\dim^{\widetilde{\varphi}},\infty\big).\]
		To verify the first inclusion, note that, by Observation~\ref{obs:36},
		\begin{align*}
			t>s>\dim^{\widetilde{\varphi}}(x)
			&\implies \liminf_{\delta\to 0^+}2^{\C_\delta(x)}\widetilde{\varphi}_s(\delta)=0\\
			&\iff \liminf_{\delta\to 0^+}\frac{\C_\delta(x)\varphi_s(\delta)-1}{\varphi_s(\delta)}=-\infty\\
			&\implies\liminf_{\delta\to 0^+}\C_\delta(x)\varphi_s(\delta)<1\\
			&\implies t\in S^-.
		\end{align*}
		To verify the second inclusion, note that, by Observation~\ref{obs:36},
		\begin{align*}
			s\in S^-&\iff \liminf_{\delta\to 0^+}\C_\delta(x)\varphi_s(\delta)=0\\
			&\implies\liminf_{\delta\to 0^+}\frac{\C_\delta(x)\varphi_s(\delta)-1}{\varphi_s(\delta)}=-\infty\\
			&\iff \liminf_{\delta\to 0^+}\log\big(2^{\C_\delta(x)}\widetilde{\varphi}_s(\delta)\big)=-\infty\\
			&\iff \liminf_{\delta\to 0^+}2^{\C_\delta(x)}\widetilde{\varphi}_s(\delta)=0\\
			&\implies s\geq \dim^{\widetilde{\varphi}}(x).
		\end{align*}
		\item It suffices to show that
		\[\big(\Dim^{\widetilde{\varphi}},\infty\big)\subseteq S^+\subseteq \big[\Dim^{\widetilde{\varphi}},\infty\big).\]
		The proof of this is completely analogous to the proof of part 1 of the theorem.
	\end{enumerate}
\end{proof}

\section{Proofs from Section 4}
\begin{proof}[Proof of Theorem~\ref{theo31}]
	Let $X$, $\varphi$, and $E$ be as given. A function $f:\{0,1\}^*\to X$ such that $D=\range(f)$ is dense in $X$ is an implicit oracle in all Kolmogorov complexities and algorithmic dimensions in this proof, but we omit $f$ from the notation.

	For any $s\in\Q^+$, the density of $D$ implies that $H^{\varphi_{s}}(E)=0$ can be witnessed by balls with rational radii and centers in $D$. Hence, for every $s,\varepsilon\in\Q^+$ with $s>\dimh^\varphi(E)$, there exist sequences $\{x^{s,\varepsilon}_i\}_{i\in\N}\subseteq D$ and
	$\{\delta^{s,\varepsilon}_i\}_{i\in\N}\subseteq\Q^+$ such that $\{B_{\delta^{s,\varepsilon}_i}(x^{s,\varepsilon}_i)\}_{i\in\N}$
	is an $\varepsilon$-cover of $E$ and
	\begin{equation}\label{eq:coversumbound}
		\sum_{i\in\N} \varphi_{s}(\delta^{s,\varepsilon}_i)<1.
	\end{equation}

	Let $h:\N\times(\Q\cap(\dimh^\varphi(E),\infty))\times\Q^+\to \binary\times\Q$ be such that
	\[h(i,s,\varepsilon)=(w_i^{s,\varepsilon}, \delta_i^{s,\varepsilon}),\]
	where $f(w_i^{s,\varepsilon})=x_i^{s,\varepsilon}$, and let $g:(\Q^+)^3\to\R$ be the (continuous) function
	\[g(s,t,\varepsilon)=\inf_{\delta\leq \varepsilon}\log\sqrt[4]{\frac{\varphi_s(\delta)}{\varphi_{t}(\delta)}}.\]
	Let $A$ be an oracle encoding $h$ and $g$, and let $s,t\in\Q^+$ such that $\dimh^\varphi(E)<s<t$. We will show that for every $x\in E$, $\dim^{\varphi,A}(x)\le t$.

Fix $x\in E$ and $\varepsilon\in\Q^+$ such that $\C^A(\varepsilon)\leq g(s,t,\varepsilon)$; such an $\varepsilon$ exists because $g(s,t,\cdot)$ is unbounded and computable relative to $A$ (cf. Section 3.3 of~\cite{LiVit19}).
Let $j$ be such that $x\in B_{\delta_j^{s,\varepsilon}}(x_j^{s,\varepsilon})$, and let $\delta=\delta_j^{s,\varepsilon}\leq\varepsilon$.
Then as $\delta\to 0^+$,
\begin{align*}
	\C_\delta^{A}(x)&\leq\C^A(x_j^{s,\varepsilon})+O(1)\\
	&\leq \C^A(\varepsilon,j)+O(1).
\end{align*}
By~\eqref{eq:coversumbound} there are fewer than $1/\varphi_s(\delta)$
values of $i$ for which $\delta_i^{s,\varepsilon}=\delta$. Therefore
\begin{align*}
	\C^A(\varepsilon,j)&\leq 2\C^A(\varepsilon) + \C^A(j)+O(1)\\
	&\leq 2 g(s,t,\varepsilon)+\log\frac{1}{\varphi_s(\delta)}+O(1)\\
	&\leq 2\log \sqrt[4]{\frac{\varphi_s(\delta)}{\varphi_{t}(\delta)}}+\log\frac{1}{\varphi_s(\delta)}+O(1)\\
	&= \log\sqrt{\frac{\varphi_s(\delta)}{\varphi_{t}(\delta)}}+\log\frac{1}{\varphi_s(\delta)}+O(1).
\end{align*}
Since we can choose arbitrarily small $\varepsilon$ in the above analysis, we have shown
\begin{align*}
	\liminf_{\delta\to 0^+} 2^{\C_\delta^{A}(x)}\varphi_{t}(\delta)&\leq \liminf_{\delta\to 0^+} 2^{\log\sqrt{\frac{\varphi_s(\delta)}{\varphi_{t}(\delta)}}+\log\frac{1}{\varphi_s(\delta)}+O(1)}\varphi_{t}(\delta)\\
	&\leq O(1)\cdot\liminf_{\delta\to 0^+}\sqrt{\frac{\varphi_t(\delta)}{\varphi_{s}(\delta)}}\\
	&=0,
\end{align*}
and so $\dim^{\varphi,A}(x)\le t$.

For the other direction, assume that there is a precision family $\alpha=\{\alpha_s\}_{s\in(0,\infty)}$ for $\varphi$, fix any oracle $A\subseteq\N$, and let $s,t\in\Q$ be such that
\begin{equation}\label{eq:dimss}
	\sup_{x\in E}\dim^{\varphi,A}(x)<s<t.
\end{equation}
For all $r\in \N$, let
\[\mathcal{U}_r=\left\{B_{\alpha_{s}(r)}(f(w))\;\middle|\;\C^{A}(w)\le \log\frac{1}{\varphi_{s}(\alpha_{s}(r))}\right\},\]
and notice that
\[|\mathcal{U}_r|\leq\frac{2}{\varphi_{s}(\alpha_{s}(r))}.\]
Now fix any $r\in\N$, and let
\[\mathcal{W}_r=\bigcup_{k=r}^\infty \mathcal{U}_k.\]
For every $x\in E$,~\eqref{eq:dimss}, together with the fact that $\varphi(\alpha_s(r))=O(\varphi(\alpha_s(r+1)))$ as $r\to\infty$, tells us that the set 
\[\left\{r\;\middle|\; 2^{\C_{\alpha_s(r)}^{A}(x)}\varphi_{s}(\alpha_s(r))\le 1\right\}=\left\{r\;\middle|\; \C_{\alpha_s(r)}^{A}(x)\leq\log\frac{1}{\varphi_{s}(\alpha_s(r))}\right\}\]
is unbounded, as is the set $\{k\mid x\in\mathcal{U}_k\}$, so $x\in\mathcal{W}_r$. Thus $\mathcal{W}_r$ is a countable $\alpha_s(r)$-cover of $E$ with
\begin{align*}
	\sum_{U\in \mathcal{W}_r} \varphi_t(\diam(U))
	&=\sum_{k=r}^\infty \sum_{U\in \mathcal{U}_k} \varphi_t(\alpha_s(k))\\
	&< \sum_{k=r}^\infty \frac{2}{\varphi_{s}(\alpha_s(k))}\cdot\varphi_t(\alpha_s(k)).
\end{align*}
Since $\alpha$ is a precision family for $\varphi$, this sum converges. We thus have
\begin{align*}
	H^{\varphi_t}(E)&\leq 2\lim_{r\to\infty}\sum_{k=r}^\infty \frac{\varphi_t(\alpha_s(k))}{\varphi_s(\alpha_s(k))}\\
	&\leq 2\sum_{r\in\N}\frac{\varphi_t(\alpha_s(r))}{\varphi_s(\alpha_s(r))}\\
	&<\infty,
\end{align*}
so $\dimh^{\varphi}(E)\le t$. We conclude
that $\dimh^{\varphi}(E)\le\sup_{x\in E}\dim^{\varphi,A}(x)$.
\end{proof}

\section{Proofs from Section 5}
\begin{construction}\label{cons41} Given a sequence
$R\in\bin^{\omega}$, define a sequence $A_0, A_1, \ldots $ of
$2^\ell$-element sets $A_\ell\subseteq\bin^{2\ell}$ by the following
recursion.
\begin{enumerate}[(i)]
\item $A_0=\{\lambda\}$, where $\lambda$ is the empty string.
\item Assume that \[A_\ell=\myset{u_i}{0\le i<2^\ell}\subseteq
\bin^{2\ell}\] has been defined, where the $u_i$ are in lexicographical
order. Let \[b_{u_10}, b_{u_11}, b_{u_20}, b_{u_21}, \ldots,
b_{u_{2^\ell}0}, b_{u_{2^\ell}1}\] be the first $2^{\ell+1}$ bits of $R$ that
have not been used in earlier stages of this construction. Then
\[A_{\ell+1}=\myset{uab_{ua}}{u\in A_\ell\mbox{ and }a\in\bin}.\]
\end{enumerate}
For each string $w\in\binary$ of even length, define the closed
interval $I_w\subseteq[0,1]$ of length $7^{-|w|/2}$ by the following
recursion.
\begin{enumerate}[(i)]
\item $I_{\lambda}=[0,1]$.
\item Assume that $I_w$ has been defined, and divide $I_w$ into
seven equal-length closed intervals, calling these $K_1$, $J_{00}$,
$J_{01}$, $K_2$, $J_{10}$, $J_{11}$, and $K_3$, from left to right. Then, for each
$a,b\in\bin$, we have $I_{wab}=J_{ab}$.
\end{enumerate}
For each $\ell\in\N$, let \[E_\ell=\bigcup_{w\in A_\ell}I_w,\] and let
\[E=\bigcap_{\ell=0}^{\infty}E_\ell.\] This completes the construction.
\end{construction}

Intuitively, each $E_\ell$ in Construction \ref{cons41} consists of
$2^\ell$ closed intervals of length $7^{-\ell}$, with gaps between them of
length at least $7^{-\ell}$. For each of these, one bit of $R$ decides
which of the subintervals $J_{00}$ and $J_{01}$ is included in
$E_{\ell+1}$, and the next bit of $R$ decides which of the subintervals $J_{10}$ and $J_{11}$ is included in $E_{\ell+1}$. The set $E$ is a Cantor-like set chosen in this fashion.
It is clear that $E$ is compact.

\begin{observation}\label{obs42} For all $R\in\bin^{\omega}$, the
set $E$ of Construction \ref{cons41}\ has Hausdorff and packing dimensions
\[\dimh(E)=\dimpack(E)= \frac{\log 2}{\log 7} \approx 0.356.\]
\end{observation}
\begin{proof}
	Let $R\in\{0,1\}^\omega$, and let $R_0=0^\omega$. Let $E$ be the set constructed from $R$ as in Construction~\ref{cons41}, and let $E_0$ be the set constructed from $R_0$. Note that $E_0$ is the set of all reals in $[0,1]$ whose base-7 expansions consist entirely of 1s and 4s. Define the function
	\[g:\{1,4\}^\omega\to\{1,2,4,5\}^\omega\]
	by
	\[g(S)[n]=S[n]+R\big[2n+S[n]\big]\]
	for all $S\in\{1,4\}^\omega$ and $n\in\N$. Note that $g$ transforms each 1 in $S$ to a 1 or 2 in $g(S)$, and $g$ transforms each 4 in $S$ to a 4 or 5 in $g(S)$. If we identify sequences in $\{1,2,4,5\}$ with the reals that they represent in base 7, then we now have a bijection
	\[g:E_0\xrightarrow[\text{onto}]{\text{1-1}}E.\]
	Moreover, if $x,y\in E_0$ are distinct, and $n$ is the first position at which $x$ and $y$ have different base-7 digits, then
	\[|x-y|,\,|g(x)-g(y)|\in\left[7^{-n},7^{1-n}\right],\]
	so $g$ is bi-Lipschitz and hence preserves Hausdorff and packing dimensions~\cite{Falc14}. We thus have
	\begin{equation}\label{eq:e0e}
		\dimh(E)=\dimh(E_0),\ \dimpack(E)=\dimpack(E_0).
	\end{equation}
	The set $E_0$ is the self-similar fractal given by an iterated function system consisting of two contractions, each with ratio $\frac{1}{7}$. It follows by the fundamental theorem on self-similar fractals~\cite{Mora46,Falc89,Falc14} that the Hausdorff and packing dimensions of $E_0$ are both the unique solution $s$ of the equation $2\cdot 7^{-s}=1$, i.e., that
	\begin{equation}\label{eq:e0}
		\dimh(E_0)=\dimpack(E_0)=\frac{\log 2}{\log 7}.
	\end{equation}
	The observation follows from~\eqref{eq:e0e} and~\eqref{eq:e0}.
\end{proof}

\begin{observation}\label{obs43} If $R\in\bin^{\omega}$ is Martin-L\"of random, then the
set $E$ of Construction \ref{cons41}\ has algorithmic dimensions
\[\dim(E)=\Dim(E)= \infty.\]
\end{observation}
\begin{proof}[Proof of Observation \ref{obs43}\ (sketch)]

Let $R$ and $E$ be as given. By~\eqref{eq21} it suffices to show
that, for all sufficiently large $r$, \[\C_r(E)>2^{r/3}.\] For this
is suffices to show that, for all sufficiently large $r$ and all
$F\in\cD$, \begin{equation}\label{eqap1}\rhoH(F, E)\le 2^{-r}
\implies \C(F)>2^{r/3}.\end{equation}

Let $r\in\N$ be large, and assume that $\rhoH(F,E)\le 2^{-r}$. Let
$\ell=\lceil r/3\rceil$. Then $2^{-r}< 1/2\cdot 7^{-\ell}$, so the finite
set $F$ can be used to compute the set $A_l$ of Construction
\ref{cons41}. This implies that $F$ can be used to compute the
$(2^\ell-1)$-bit prefix $w$ of $R$ that was used to decide the set
$A_\ell$. Since $R$ is random and $r$ is large, this implies that
$\C(F)>2^{r/3}$.
\end{proof}

In addition to illustrating the difference between classical and algorithmic dimensions, Observation \ref{obs43}\ combines with our general
point-to-set principle to give a very non-classical proof of the
following known classical fact.

\begin{corollary}\label{cor44}
	$\dimh(\cK([0,1]))=\dimpack(\cK([0,1])=\infty$.
\end{corollary}

\begin{proof}
Let $A\subseteq\N$. By Theorem \ref{theo31}\ applied to
$\cK([0,1])$, it suffices to exhibit a point $E\in\cK([0,1])$ such
that $\dim^A(E)=\infty$. If we choose $R\in\bin^{\omega}$ to be
Martin-L\"of random relative to $A$, then Observation \ref{obs43},
relativized to $A$, tells us that Construction \ref{cons41}\ gives
us just such a point.
\end{proof}

\begin{proof}[Proof of Theorem~\ref{thm:hypmink}]
	Let $E\subseteq X$ and $\varphi$ be a gauge family. Let $\delta>0$ and $F\subseteq X$ be such that $|F|=N(E,\delta)$ and
	\[E\subseteq \bigcup_{x\in F}B_\delta(x).\]
	For every $L\in \cK(E)$, we have $\rhoH(L,\{x\in F\mid B_\delta(x)\cap L\neq \emptyset\})<\delta$, so
	\[\cK(E)\subseteq\bigcup_{T\subseteq F}B_\delta(T),\]
	and therefore $N(\cK(E),\delta)\leq |2^F|=2^{N(E,\delta)}$.

	Now suppose that $\liminf_{\delta\to 0^+}N(E,\delta)\varphi_s(\delta)=0$. Then since $\varphi_s(\delta)\to 0^+$ as $\delta\to 0^+$, we have
	\begin{align*}
		\liminf_{\delta\to 0^+}N(\cK(E),\delta)\widetilde{\varphi}_s(\delta)&\leq \liminf_{\delta\to 0^+} 2^{N(E,\delta)}2^{-1/\varphi_s(\delta)}\\
		&=\liminf_{\delta\to 0^+}2^{\frac{N(E,\delta)\varphi_s(\delta)-1}{\varphi_s(\delta)}}\\
		&=0,
	\end{align*}
	so $\dimm^{\widetilde{\varphi}}(\cK(E))\leq \dimm^\varphi(E)$.

	We now show that $\dimm^{\widetilde{\varphi}}(\cK(E))\geq \dimm^\varphi(E)$. Let $\delta>0$, let $P$ be a set of $2\delta$-separated points in $E$, and observe that $|P|\leq N(E,\delta)$. Let $\mathcal{F}\subseteq \cK(E)$ satisfy $|\mathcal{F}|=N(\cK(E),\delta)$ and
	\[\cK(E)\subseteq\bigcup_{F\in \mathcal{F}}B_\delta(F).\]
	For every distinct pair $S,S'\subseteq P$, we have $\rhoH(S,S')\geq 2\delta$, so, for each $F\in\mathcal{F}$, the ball $B_\delta(F)$ can contain at most one subset of $P$. Hence, \[N(\cK(E),\delta)=|\mathcal{F}|\geq |2^P|\geq 2^{N(E,\delta)}.\]
	Now let $t>s>\dimm^{\widetilde{\varphi}}(\cK(E))$. Then
	\begin{align*}
		\liminf_{\delta\to 0^+}N(\cK(E),\delta)\widetilde{\varphi}_s(\delta)=0
		&\implies \liminf_{\delta\to 0^+}2^{N(E,\delta)}2^{-1/\varphi_s(\delta)}=0\\
		&\iff\liminf_{\delta\to 0^+}\left(N(E,\delta)-\frac{1}{\varphi_s(\delta)}\right)=-\infty\\
		&\iff \liminf_{\delta\to 0^+}\frac{N(E,\delta)\varphi_s(\delta)-1}{\varphi_s(\delta)}=-\infty\\
		&\implies \liminf_{\delta\to 0^+}N(E,\delta)\varphi_s(\delta)<1\\
		&\implies \liminf_{\delta\to 0^+} N(E,\delta)\varphi_t(\delta)=0\\
		&\implies \dimm^{\varphi}(E)\leq t.
	\end{align*}
	The argument for upper Minkowski dimension is completely analogous.
\end{proof}

\begin{observation}\label{obs:weakmdimball}
	Let $\varphi$ be any gauge family, $X$ any metric space, $E\subseteq X$, and $\delta>0$.
	\begin{enumerate}
		\item If $\dimm^\varphi(E)<\infty$, then there exists a point $x\in E$ such that $\dimm^\varphi(E\cap B_\delta(x))=\dimm^\varphi(E)$.
		\item If $\Dimm^\varphi(E)<\infty$, then there exists a point $x\in E$ such that $\Dimm^\varphi(E\cap B_\delta(x))=\Dimm^\varphi(E)$.
	\end{enumerate}
\end{observation}

\begin{proof}
	It follows from the monotonicity of Minkowski dimensions that
	\[\dimm^\varphi(E\cap B_\delta(x))\leq \dimm^\varphi(E)\]
 	holds for every $x\in X$.
	
	If $\dimm^\varphi(E)<\infty$, then $N(E,\delta/2)$ is finite; let the set $\{x_1,\ldots,x_{N(E,\delta/2)}\}\subseteq X$ testify to this value. Then every $E\cap B_{\delta/2}(x_i)$ is nonempty, and
	\[E\subseteq \bigcup_{i=1}^{N(E,\delta/2)} (E\cap B_{\delta/2}(x_i)),\]
	so by the finite stability of $\dimm^\varphi$, there is some $1\leq i\leq N(E,\delta/2)$ such that
	\[\dimm^\varphi(E)\leq \dimm^\varphi(E\cap B_{\delta/2}(x_i)).\]
	Now let $x\in E\cap B_{\delta/2}(x_i)$ and observe that $E\cap B_{\delta/2}(x_i)\subseteq E\cap B_\delta(x)$, so monotonicity gives
	\[\dimm^\varphi(E\cap B_{\delta/2}(x_i))\leq \dimm^\varphi(E\cap B_{\delta}(x)).\]
	This proves the first statement, and the proof of the second statement is completely analogous.
\end{proof}

\begin{lemma}\label{lem:strongmdimball}
	Let $\varphi$ be any gauge family, $X$ any metric space, and $E\subseteq X$  a compact set.
	\begin{enumerate}
		\item There exists a point $x\in E$ such that $\dimm^\varphi(E\cap B_\delta(x))=\dimm^\varphi(E)$ holds for all $\delta>0$.
		\item There exists a point $x\in E$ such that $\Dimm^\varphi(E\cap B_\delta(x))=\Dimm^\varphi(E)$ holds for all $\delta>0$.
	\end{enumerate}
\end{lemma}

\begin{proof}
	By the compactness of $E$, $N(E,\delta)$ is finite for every $\delta>0$, and the Minkowski dimensions of $E$ are finite as well. Hence, Observation~\ref{obs:weakmdimball} yields a sequence $\{x_r\}_{r\in\N}$ of points in $X$ such that
	\[\dimm^\varphi(E\cap B_{2^{-r}}(x_r))=\dimm^\varphi(E)\]
	for all $r\in\N$. Since $E$ is compact, there is a subsequence $\{x_{r_i}\}_{i\in\N}$ of $\{x_r\}_{r\in\N}$ that converges to some point $x\in E$. Thus, for all $\delta>0$, there is an $i\in\N$ such that $\rho(x_{r_i},x)<2^{-r_i}<\delta/2$, so \[B_{2^{-r_i}}(x_{r_i})\subseteq B_{2^{1-r_i}}(x)\subseteq B_\delta(x).\]
	By the monotonicity of Minkowski dimensions, then,
	\[\dimm^\varphi(E)= \dimm^\varphi(E\cap B_{2^{-r_i}}(x_{r_i}))\leq \dimm^\varphi(E\cap B_{\delta}(x))\leq\dimm^\varphi(E).\]
	The proof of the second statement is completely analogous.
\end{proof}

\begin{proof}[Proof of Theorem~\ref{thm:packminkequiv}]
	Lemma~\ref{lem:tricot} immediately gives $\dimpack^{\widetilde{\varphi}}(\cK(E))\leq\Dimm^{\widetilde{\varphi}}(\cK(E))$.

	For the other direction, apply the general point-to-set principle for packing dimension (Theorem~\ref{the32}) to let $A$ be an oracle such that
	\begin{equation}\label{eq:p2spackmink}
		\dimpack^{\widetilde{\varphi}}(\cK(E))\geq \sup_{L\in\cK(E)}\Dim^{\widetilde{\varphi},A}(L),
	\end{equation}
	and let $t=\Dimm^{\varphi}(E)$. Applying Lemma~\ref{lem:strongmdimball}, let $x\in E$ be a point such that for all $\delta>0$
	\[\Dimm^{\varphi}(E\cap B_{\delta}(x))=t.\]
	By the hyperspace Minkowski dimension theorem, then, we also have
	\begin{equation}\label{eq:Dimintersectball}
		\Dimm^{\widetilde{\varphi}}(\cK(E\cap B_{\delta}(x)))=t
	\end{equation}
	for all $\delta>0$.

	Let $s<t$. We will recursively define a compact set $L\in\cK(E)$ such that
	\[\Dim^{\widetilde{\varphi},A}(L)> s.\]
	Let $A_0$ be an oracle that encodes both $A$ and $x$. By~\eqref{eq:Dimintersectball} and Observation~\ref{obs:separabledim},
	\[\limsup_{\delta\to 0^+}\hat{N}(\cK(E\cap B_1(x)),\delta)\widetilde{\varphi}_s(\delta)=\infty.\]
	Thus there is a precision $\delta_0\in \Q^+$ such that
	\[\hat{N}(\cK(E\cap B_1(x)),\delta_0)\widetilde{\varphi}_s(\delta_0)>1.\]
	That is, it requires at least $1/\widetilde{\varphi}_s(\delta_0)$ open balls of radius $\delta_0$ (in the $\rhoH$ metric), with centers that are finite subsets of $D$, to cover $\cK(E\cap B_1(x))$. By the pigeonhole principle, the number of finite sets $J\subseteq D$ satisfying 
	\[\C^{A_0}(J)\leq \log\left(\frac{1}{2\widetilde{\varphi}_s(\delta_0)}\right)\]
	is at most
	\[2^{\log\left(\frac{1}{2\widetilde{\varphi}_s(\delta_0)}\right)+1}-1<\frac{1}{\widetilde{\varphi}_s(\delta_0)}.\]
	Hence there is some compact set $L_0\in\cK(B_1(x)\cap E)$ with
	\begin{align*}
		\C^{A_0}_{\delta_0}(L_0)\varphi_s(\delta_0)
		&>\log\left(\frac{1}{2\widetilde{\varphi}_s(\delta_0)}\right)\varphi_s(\delta_0)\\
		&=1-\varphi_s(\delta_0)\\
		&>1/2.
	\end{align*}	
	Define the compact set
	\[L_0'=(L_0\setminus B_{\delta_0}(x))\cup\{x\},\]
	and notice that $\rhoH(L_0\cup\{x\},L_0')\leq \delta_0$, so
	\[\C^{A_0}_{\delta_0}(L_0)\leq \C^{A_0}_{\delta_0}(L_0,x)+O(1)\leq\C^{A_0}_{\delta_0}(L'_0)+O(1),\]
	since $A_0$ encodes $x$. Thus, as long as $\delta_0$ is sufficiently small, we have 
	\[\C^{A_0}_{\delta_0}(L_0')\varphi_s(\delta_0)\geq 1/4.\]

	Now, for each $i\geq 1$, let $A_i$ be an oracle encoding $A_{i-1}$ and $\delta_{i-1}$. Let $\delta_i\in\Q^+$ and
	\[L'_i\in \cK\left(B_{\delta_{i-1}/2}(x)\cap E\right)\]
	be such that $L_i'\cap B_{\delta_i}(x)=\{x\}$ and
	\begin{equation}\label{eq:liprimepack}
		\C^{A_i}_{\delta_i}(L_i')\varphi_s(\delta_i)\geq 1/4.
	\end{equation}
	This pair exists for exactly the same reason that $\delta_0$ and $L_0'$ exist.

	Define the set
	\[L=\bigcup_{i\in\N} L_i'.\]
	Notice first that this set belongs to $\cK(E)$. Consider any sequence $\{x_r\}_{r\in\N}$ of points in $L$. If the sequence is contained within $\{x\}\cup\bigcup_{i=0}^n L'_i$ for some finite $n$---i.e., within a finite union of compact sets, which is compact---then it has a convergent subsequence that converges to a point in that union. Otherwise, the sequence has points in infinitely many of the $L'_i$, and there is a subsequence $\{x_{r_j}\}_{j\in\N}$ such that, for every pair $j'>j$ there exists a pair $i'>i$ such that $x_{r_{j}}\in L'_i\setminus\{x\}$ and $x_{r_{j'}}\in L'_{i'}\setminus\{x\}$; such a subsequence converges to $x$. Thus $L$ is sequentially compact and therefore compact.

	Let $D$ be a countable dense set in $X$, and recall that $U$ is a fixed universal oracle prefix Turing machine. Consider an oracle prefix Turing machine $M$, with access to an oracle for $x$ and $\delta_i$. On input $\pi$ such that $U(\pi)=F\subseteq D$, $M$ outputs the set
	\[\{y\in F\mid \rho(x,y)<\delta_i/2\}.\]
	Now let $\pi$ testify to $\C_{\delta_i}^{A_i}(L)$. Then $M(\pi)\subseteq D$ is a set of points satisfying $\rhoH(L_i',U(\pi))< \delta_i$, so we have
	\begin{equation}\label{eq:optconstpack}
		\C^{A_i}_{\delta_i}(L_i')\leq |\pi|+c_M=\C_{\delta_i}^{A_i}(L)+c_M\,
	\end{equation}
	where $c_M$ is an optimality constant for the machine $M$. Furthermore, \[\C_{\delta_i}^{A_i}(L)\leq \C_{\delta_i}^{A}(L)+O(1).\]
	Combining this fact with~\eqref{eq:liprimepack} and~\eqref{eq:optconstpack} yields
	\[\C_{\delta_i}^{A}(L)\varphi_s(\delta_i)\geq 1/4-\varphi_s(\delta_i)\cdot O(1).\]
	The latter term vanishes as $i\to\infty$, so
	\[\limsup_{\delta\to 0^+}\C_{\delta}^{A}(L)\varphi_s(\delta)\geq 1/4.\]
	By Theorem~\ref{thm:DimL}, this implies that $\Dim^{\widetilde{\varphi},A}(L)> s$. We conclude that $\Dim^{\widetilde{\varphi},A}(L)\geq t$, so by~\eqref{eq:p2spackmink}, the proof is complete.
\end{proof}

\end{document}